\documentclass[a4paper,10pt,oneside,reqno]{article}
\newcommand*{\affaddr}[1]{#1} 
\newcommand*{\affmark}[1][*]{\textsuperscript{#1}}

\usepackage{mathtools, amsfonts, mathrsfs}
\usepackage{subcaption}
\usepackage{graphicx}
\usepackage{booktabs}
\usepackage{threeparttable}
\usepackage{lineno}
\usepackage[toc]{appendix}
\usepackage{amssymb}
\usepackage{color}
\usepackage{hypernat}
\usepackage[a4paper, total={6in, 8in}, margin=1.in, bottom=1in]{geometry}
\usepackage{cite}
\usepackage{amsmath}
\usepackage{amssymb}
\usepackage{mathtools}
\usepackage{mathrsfs}
\usepackage{amsthm}
\usepackage{chngcntr}
\usepackage{apptools}

\AtAppendix{\counterwithin{theorem}{section}}
 
\usepackage{nomencl}
\makenomenclature

\renewcommand{\nompreamble}{The next list describes several symbols that will be later used within the body of the document}

\usepackage[ampersand]{easylist}
\ListProperties(Hide=100, Hang=true, Progressive=3ex, Style*=-- ,Style2*=$\bullet$ ,Style3*=$\circ$ ,Style4*=\tiny$\blacksquare$ )
\modulolinenumbers[1]

\usepackage{tikz}
\usetikzlibrary{matrix} 
\usetikzlibrary{arrows} 
\usetikzlibrary{calc} 

\tikzstyle{block} = [draw,rectangle,thick,minimum height=2em,minimum width=2em]
\tikzstyle{connector} = [->,thick]

\numberwithin{equation}{section}

\theoremstyle{plain}
  \newtheorem{theorem}{Theorem}[section]

  \newtheorem{corollary}[theorem]{Corollary}

\theoremstyle{definition}
  
  \newtheorem{example}[theorem]{Example}
  \newtheorem{remark}[theorem]{Remark}

\newcommand{\R}{\mathbb{R}}

\newcommand{\N}{\mathbb{N}}

\newcommand{\Var}{\mathrm{Var}}
\newcommand{\Cov}{\mathrm{Cov}}
\newcommand{\E}{\mathbb{E}}
\newcommand{\eps}{\varepsilon}

\title{Sensitivity analysis based dimension reduction of\\multiscale models}
\author{
Anna Nikishova\affmark[a], Giovanni Eugenio Comi\affmark[b], Alfons G. Hoekstra\affmark[a]\\
\small \affaddr{\affmark[a]Computational Science Lab, Institute for Informatics, Faculty of Science,\\ \small University of Amsterdam, The Netherlands}\\
\small \affaddr{\affmark[b]Scuola Normale Superiore di Pisa, Italy}
}

\begin{document}

\maketitle



\begin{abstract}
In this paper, the sensitivity analysis of a single scale model is employed in order to reduce the input dimensionality of the related multiscale model, in this way, improving the efficiency of its uncertainty estimation. The approach is illustrated with two examples: a reaction model and the standard Ornstein-Uhlenbeck process. Additionally, a counterexample shows that an uncertain input should not be excluded from uncertainty quantification without estimating the response sensitivity to this parameter. In particular, an analysis of the function defining the relation between single scale components is required to understand whether single scale sensitivity analysis can be used to reduce the dimensionality of the overall multiscale model input space.
\end{abstract}

\section{Introduction}
Results of computational models should be supported by uncertainty estimates whenever precise values of their inputs are not available \cite{ROY20112131,UUSITALO201524,Soize_2017}. This is usually the case since measurements of inputs rarely can be made exactly, or inputs may include aleatory uncertainty \cite{OBERKAMPF2002333,URBINA20111114}. Uncertainty Quantification (UQ) of a complex model usually requires powerful computational resources. Moreover, the cost of some UQ methods increases exponentially with the number of uncertain inputs. 

Sensitivity analysis (SA) identifies the effects of uncertainty in a model input or group of inputs to the model response. In 1990, Sobol introduced sensitivity indices to measure the effect of input uncertainty on the model output variance \cite{Sob90, Sob107}. In \cite{sobol2001global,SOBOL2007957}, Sobol employs SA in order to fix uncertain parameters with low total sensitivity indices and reduce the model dimensionality. 

Here such application of SA to multiscale models is considered. A multiscale model is defined as a collection of single scale models that are coupled using scale bridging methods. The approach proposed here consists in examining the type of function coupling the single scale components, followed by estimating the sensitivity of the response of a single scale model. This paper demonstrates that estimates of the single scale model sensitivity can be used to assess the sensitivity of the overall multiscale model response for some classes of multiscale model functions. However, this is not always possible, as will be shown by a counterexample.

Sobol's variance based approach is the preferred method to measure model output sensitivity \cite{SOBOL20093009,KUCHERENKO2009,saltelli2010avoid,SALTELLI201929}. Even though it is important to note that variance is not always the most representative measure of model response uncertainty \cite{BORGONOVO2016869,KUCHERENKO201935}, it is assumed to be so in this work. The proposed approach is based on exploring the coupled structure of multiscale models, allowing to analyse independently the single scale models. Therefore, the second assumption is that SA can be performed on the multiscale model components. Additionally, it is assumed that the multiscale model parameters are uncorrelated.

In Section \ref{sec:MSmodels}, a brief description of multiscale models is given. Section~\ref{sec:sa} is devoted to SA, and its application to dimensionality reduction of a multiscale model is discussed in subsection~\ref{sec:SA}. Together with some examples of the sensitivity analysis for multiscale models (subsections~\ref{sec:case_1} and~\ref{sec:case_2}), a counterexample is considered in subsection~\ref{sec:Counterexamples} in order to illustrate that, even though it is tempting to employ the SA result of single scale models to the response of the overall multiscale model, this is not always allowed. Section~\ref{sec:Conclusions} summarizes the results and includes a note on the application of the proposed approaches to some real-world models. Some other cases of multiscale models for which the proposed method on dimension reduction can be applied are in the Appendix. In particular, in the \ref{sec:general_res} an upper bound for the sensitivity of model output for a general class of coupling function is obtained.

\section{Multiscale model}\label{sec:MSmodels}

Following the concept introduced in the Multiscale Modelling and Simulation Framework (MMSF) \cite{Borgdorff_2013,borgdorff2014performance,Chopard_Falcone_Kunzli_Veen_Hoekstra_2018}, multiscale models are considered as a set of single scale models coupled using scale bridging methods. 
The single scale models represent processes that operate on well defined spatio-temporal scales. In MMSF, the single scale models are placed on a scale separation map (SSM), where axes indicate the spatial-temporal scales. An example of SSM with a multiscale model that consists of two single scale components is shown in Figure~\ref{fig:MSmodel}. The directed edges between the single scale components indicate their interactions. In general, cyclic and acyclic coupling topologies are recognised: the cyclic one, as in Figure~\ref{fig:MSmodel}, assumes a feedback loop between the components, and in the acyclic one, no feedback is present. Here we rely on the assumptions of a component-based structure of the multiscale models as well as on a drastic difference in the computational cost of the single scale components.

\begin{figure}
\centering
\footnotesize
\tikzstyle{vertex}=[draw, minimum width=40pt, minimum height=40pt, align=center]
\begin{tikzpicture}[transform shape]
\draw[gray, thick, ->] (0.75,-2.3) -- (0.75, 2.);
\draw[gray, thick, ->] (0.2,-2.) -- (6.2, -2.);
\node[rotate=90]  at (.5, 1.) {Spatial scale};
\node[]  at (5., -2.3) {Temporal scale};
\node(inputs_f) at (3, 1.) {\begin{tabular}{cc}$x$\end{tabular}};
\node(inputs_h) at (1.1, -1) {\begin{tabular}{cc}$\xi$\end{tabular}};
\node[vertex](f) at (4.5, 1.) {\begin{tabular}{cc}$G(f, h)$\end{tabular}};
\node[vertex](h) at (2.5, -1.) {\begin{tabular}{cc}$h$\end{tabular}};
\draw[-latex,bend right]  (f) edge (h);
\draw[-latex,bend right] (h) edge (f);
\node(QoI) at (6.6, 1.) {\begin{tabular}{cc}$z = g(x, \xi)$\end{tabular}};
\draw[->] (f) -- node [] {} (QoI);
\draw[->] (inputs_f) -- node [] {} (f);
\draw[->] (inputs_h) -- node [] {} (h);
\end{tikzpicture}
\caption{Scale separation map. The functions $G(f, \cdot)$ and $h$ are the macro and micro models with inputs $x$ and $\xi$, respectively. The function $G(f(x), h(\xi))$ defines the relation between the response of the micro model and the rest of the macro model parameters denoted by $f$. The final multiscale model output $z = g(x, \xi) = G(f(x), h(\xi))$ is produced by the macro model.}
\label{fig:MSmodel}
\end{figure}
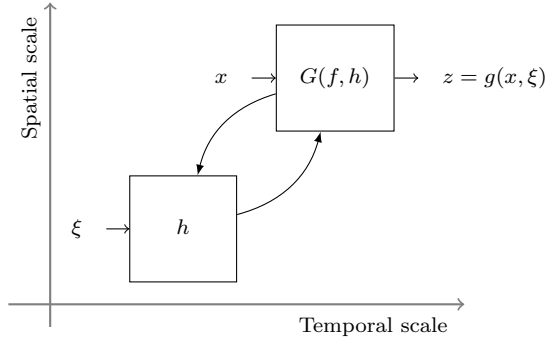

The overall multiscale model is denoted by a function $g(x, \xi) = z$ such that
$$g:\R^{n+m} \to \R^q$$
with $n, m, q \in \N$ and $\E[|g|^2] < \infty$, which produces the Quantity of Interest (QoI) $z$. We introduce a function $G: \R^{s+p} \to \R^q$, with $s, p \in \N$, as a representation of $g$, which underlines the relationship between the micro model response and the remaining variables inside the macro model, denoted by the function $f$: 
\begin{equation*}
    g(x, \xi) = G(f(x), h(\xi)).
\end{equation*}
Therefore, the function $G(f(x), \cdot)$ represents the macro model for some $f: \R^{n} \to \R^s$ which depends on parameters $x = (x_{1}, \dots, x_{n})$.
It is assumed that $f$ can be executed in a relatively short computational time, that it has a finite non-zero variance, i.e. $\E[|f|^2] < \infty$ and $f$ is not constant, and that it is possible to obtain its output sensitivity.

The micro scale component is defined by a function $h: \R^{m} \to \R^p$ which satisfies $\E[|h|^2] < \infty$. The sets of variables on which the function $h$ depends\footnote{Additionally, $h$ may depend on the macro model response. When this is the case, the micro model function is denoted by $h(x,\xi)$ or $h(x)$, meaning that it depends on the same uncertain inputs as the macro model function $f$. This is a relevant feature of the method presented here.} are of the form $\xi = (\xi_{1}, \dots, \xi_{m})$.  

Without loss of generality, later in the text it is assumed that the uncertain inputs $x$ and $\xi$ follow uniform distributions $\mathcal{U}([0, 1]^n)$ and $\mathcal{U}([0, 1]^m)$, respectively. 

\section{Sensitivity analysis}\label{sec:sa}

Sensitivity analysis identifies the effect of uncertainty in the model input parameters on the model response \cite{Sal101}. The Sobol sensitivity indices \cite{Sob90,SOBOL20093009} (SIs) are widely used to measure the response sensitivity. The total SI of an input $x_i$ for the results of the multiscale model function $g(x, \xi)=z$ is given by
\begin{align}
\begin{split}
  \label{eq:STg}
  S^g_{T_{x_i}}  
  &= \frac{\Var(z) - \Var_{x_{\sim i}, \xi} \left(\mathbb{E}_{x_i}[z|x_{\sim i}, \xi] \right) }{\Var(z)} 
  \\ 
  &= \frac{\int |g(x, \xi)|^2 dx d\xi - \int |\int g(x, \xi) dx_i|^2 dx_{\sim i} d\xi} {\int |g(x, \xi)|^2 dx d\xi  -|g_0|^2},
\end{split} 
\end{align}
where $g_0 = \mathbb{E}[g(x, \xi)]$, and the notation $x_{\sim i} = (x_{1}, \dots, x_{i - 1}, x_{i + 1}, \dots, x_{n})$ is employed \cite{HOMMA19961}. In \cite{Sob90,SOBOL2007957}, the total SIs were employed to identifying the effective dimensions of a model function and to fixing unessential variables. In particular, it was shown that, when fixing $x_i$ to a value $x^0_i$ in $[0, 1]$, the error defined by
\begin{equation*}
    \delta(x^0_i) = \frac{\int \left|g(x, \xi) - g(x_{\sim i},x^0_i, \xi)\right|^2 dx d\xi}{\Var(z)}
\end{equation*}
satisfies 
\begin{equation}\label{eq:P_error}
P \left(\delta(x^0_i) < \left(1 + \frac{1}{\eps} \right)S^g_{T_{x_i}}\right) > 1 - \eps
\end{equation}
for any $\eps > 0$. This result is applied in this work, meaning that we expect with high confidence that fixing an input with a low total sensitivity index does not produce a large error in the estimates of uncertainty. Then, this fact can be employed to reduce input dimensionality, so that UQ can be performed more efficiently. However, sensitivity indices are usually not given in advance and their estimation can be a computationally expensive task as well. 

\subsection{Sensitivity analysis of multiscale models}\label{sec:SA}

In this work, it is proposed to evaluate the response sensitivity of the computationally cheap single scale model $f$ to estimate an upper bound of the sensitivity of the multiscale model output $z$. This approach can be highly computationally efficient; however, the method does not work in general.

In order to fix uncertain inputs according to single scale model SA, it should be proved that the total sensitivity for an input $x_i$ remains small also for the output of the model $g(x, \xi)$, i.e. $S^{g}_{T_{x_{i}}} \ll 1$ given that $S^{f}_{T_{x_{i}}} \ll 1$. This cannot be assumed in general, and it depends on the form of the model function $G$. 

The first step of the proposed approach is to analyse the multiscale model function $G$, as it is shown in the following sections. In the cases, in which our method applies, the next step is to estimate numerically $S^f_{T_{x_i}}$ for $i=1, \dots, n$ by a black box method, for instance from \cite{SOBOL2007957}. Then, if it is found that  $S^{f}_{T_{x_{i}}} \ll 1$, it shall follow automatically that $S^{g}_{T_{x_{i}}} \ll 1$. Hence, according to \eqref{eq:P_error}, uncertainty can be estimated with fixed $x_i$ without producing a large error.

While the results stated below hold also for vector valued functions, using the definition of total SI given in \eqref{eq:STg}, we shall work mainly with scalar functions, in order to avoid a heavy notation. 

\subsubsection{Case 1}\label{sec:case_1}
We start by considering the homogeneous case: $G: \R^2 \to \R$, given by $G(u, v) = uv$.

\begin{theorem} \label{thm:case_1} Let $g:(0, 1)^{m + n} \to \R$ be a function in $L^{2}((0, 1)^{m + n})$ such that
\begin{equation*}
  \label{eq:case1}
  g(x, \xi) = f(x)h(\xi),
\end{equation*}
for some $f : (0, 1)^{n} \to \R$ and $h: (0, 1)^{m} \to \R$ satisfying $f \in L^{2}((0, 1)^{n})$ and $h \in L^{2}((0, 1)^{m})$. Then, we have
\begin{equation} \label{eq:STf_case1}
  S^g_{T_{x_i}} = \lambda_{f, h} S^{f}_{T_{x_{i}}},  \end{equation}
where 
\begin{equation*} 
\label{eq:lambda_f_h} 
\lambda_{f, h} = \frac{\int f(x)^{2} \, dx - f_{0}^{2}}{\int f^2(x) dx- \frac{f^2_0 h^2_0}{\int h^2(\xi) d\xi}}. \end{equation*}
In particular,  
\begin{equation} \label{eq:STf_case1_upper} S^g_{T_{x_i}} \leq S^f_{T_{x_i}}, \end{equation} 
and 
\begin{equation} \label{eq:STf_case1_lower} S^{g}_{T_{x_{i}}} \ge \left ( 1 - \frac{f_{0}^{2}}{\int f^2(x) dx} \right ) S^{f}_{T_{x_{i}}}. \end{equation}
\end{theorem}
\begin{proof}
The total SI of the input $x_i$ for the results of the model $g(x, \xi)$ is equal to
\begin{align*}
\begin{split}
  S^g_{T_{x_i}}  &= \frac{\int \int f^2(x) h^2(\xi) dx d\xi - \int \int (\int f(x) h(\xi) dx_i )^2 dx_{\sim i} d\xi} {\int \int f^2(x) h^2(\xi) dx d\xi - (f_0 h_{0})^{2}} \\ 
  &= \frac{\int f^2(x) dx - \int (\int f(x) dx_i )^2 dx_{\sim i}} {\int f^2(x) dx- \frac{f^2_0 h^2_0}{\int h^2(\xi) d\xi}}\\
  &= \frac{\int f(x)^{2} \, dx - f_{0}^{2}}{\int f^2(x) dx- \frac{f^2_0 h^2_0}{\int h^2(\xi) d\xi}} 
  \frac{\int f^2(x) dx - \int (\int f(x) dx_i )^2 dx_{\sim i}}{\int f^2(x) dx- f^2_0}\\
  &=  \frac{\int f(x)^{2} \, dx - f_{0}^{2}}{\int f^2(x) dx- \frac{f^2_0 h^2_0}{\int h^2(\xi) d\xi}} S^{f}_{T_{x_{i}}},
\end{split} 
\end{align*}
from which \eqref{eq:STf_case1} follows.

By the Cauchy-Schwarz inequality, $$h^2_0 \leq \int h^2(\xi) d\xi.$$ Therefore, $\lambda_{f, h} \le 1$, and \eqref{eq:STf_case1_upper} is obtained. 
In addition, again by Cauchy-Schwarz inequality, we get
\begin{equation*} \lambda_{f, h} \ge \frac{\int f(x)^{2} \, dx - f_{0}^{2}}{\int f^2(x) dx} = 1 - \frac{f_{0}^{2}}{\int f^2(x) dx} > 0 \end{equation*}
for any $h \in L^{2}((0, 1)^{m})$. Hence, \eqref{eq:STf_case1_lower} is obtained. 
\end{proof}

Therefore, if a low sensitivity to the parameter $x_i$ is identified by computing $S^f_{T_{x_i}}$, this parameter can be excluded from UQ of the whole multiscale model. On the other hand, inequality \eqref{eq:STf_case1_lower} means that we have a lower bound for the total SI of the input $x_{i}$ for the model $g(x, \xi) = f(x) h(\xi)$, which is independent from the choice of the function $h(\xi)$. In particular, if $x_{i}$ is an important variable for the model $f(x)$, then \eqref{eq:STf_case1_lower} implies that it cannot loose dramatically its importance in the model given by $g$.

\begin{example}[Reaction equation]

An example of Case 1 can be a reaction equation presented by an acyclic model \cite{chopard2014framework} with initial conditions provided by some function $f(x)$:
\begin{align*}
\begin{split}
\frac{\partial z(t, x, \xi)}{\partial t} &= -\psi(\xi)z(t, x, \xi), \\
z(0, x, \xi) &= f(x),
\end{split} 
\end{align*}
where $x$ and $\xi$ are uncertain model inputs. The analytical solution of the equation is
$$z(t, x, \xi) = f(x) e^{-t\psi(\xi)}.$$
Therefore, if we define $h_t(\xi) = e^{-t\psi(\xi)}$, we get 
$$z(t, x, \xi) = f(x) h_t(\xi),$$
and Theorem~\ref{thm:case_1} can be applied.

Since the proposed approach is applicable to multiscale models regardless of the complexity of $f$ and $h$, in the example, these model components are represented by the following equations:
\begin{align*}
\begin{split}
\psi(\xi) &= \xi_1^2 - \xi_2,\\
f(x) &= x_1^2 + x_1x_2x_3 + x_3^3 - x_1x_3,
\end{split} 
\end{align*}
where uncertain parameters $x$ have uniform distribution $\mathcal{U}(0.9, 1.1)$, $\xi_1$ is uniformly distributed on $[0.07, 0.09]$, and $\xi_2$ on $[0.05, 0.09]$.

Sensitivity analysis of the function $f$ results in:
\begin{align*}
\begin{split}
S^f_{T_{x_1}} &\approx 2.9 \cdot 10^{-1},\\
S^f_{T_{x_2}} &\approx 7.2 \cdot 10^{-2},\\
S^f_{T_{x_3}} &\approx 6.5 \cdot 10^{-1},
\end{split} 
\end{align*}
suggesting that the parameter $x_2$ does not significantly affect the output of the function $f$. Therefore, by Theorem~\ref{thm:case_1}, the value of this parameter can be equated to its mean when estimating uncertainty of the overall model response $z$.

Figure~\ref{fig:comparison_UQ_case1}~(a) illustrates a satisfactory match between the mean values and standard deviations obtained by sampling the results varying all the uncertain inputs and keeping the input $x_2$ equal to its mean value. Figure~\ref{fig:comparison_UQ_case1}~(c) shows that the relative error in the standard deviation does not exceed 3.5\% at any simulation time. Moreover, the resulting $p$-value of Levene's test \cite{LIM1996287} is about 0.84. Therefore, the null hypothesis that the samples are obtained from distributions with equal variances cannot be rejected.

\begin{figure}
  \centering
  \includegraphics[width=\linewidth]{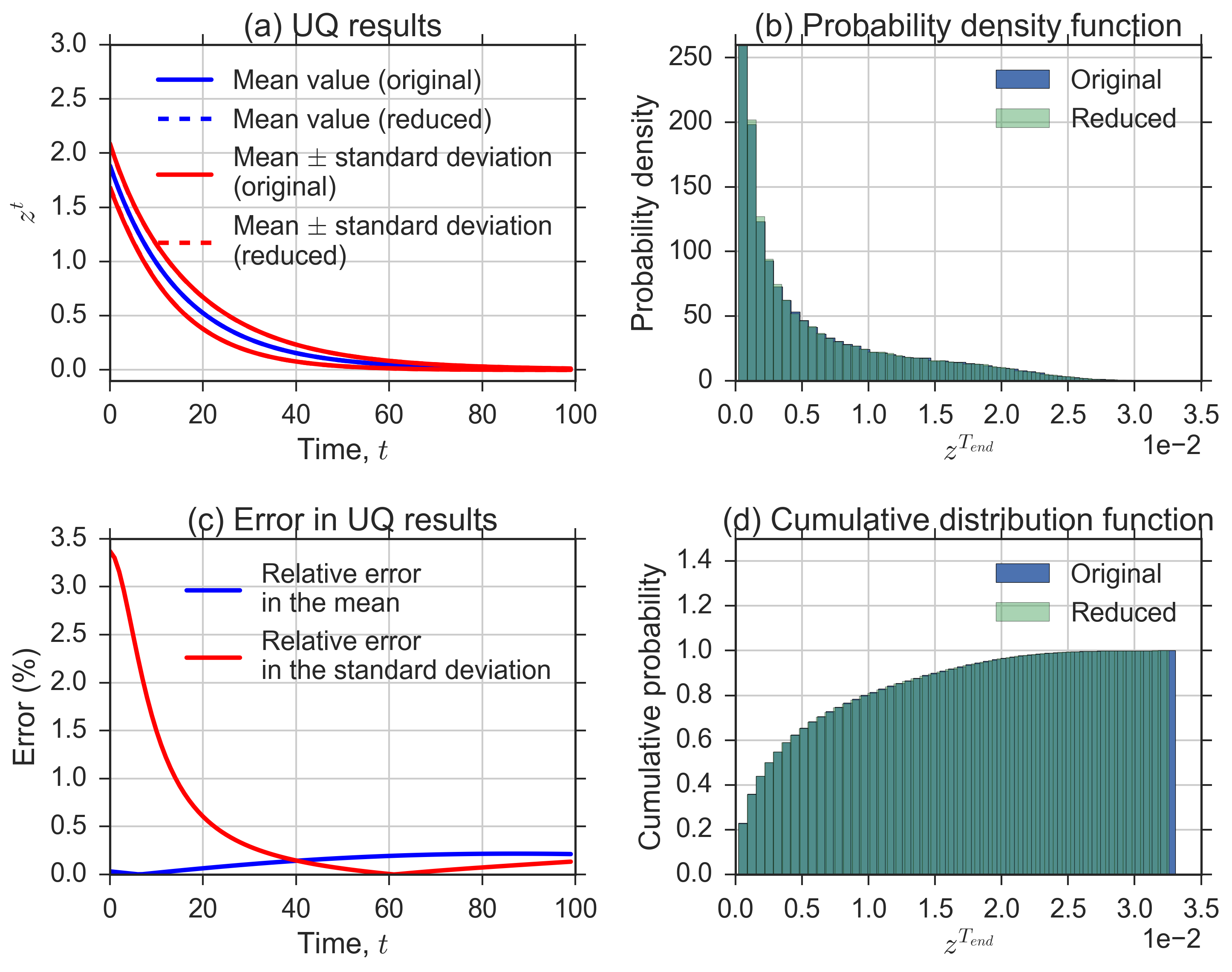}
  \caption{(a) Comparison of the estimated mean and standard deviation of the model response $z^t$ using the original sample and the sample with the unimportant parameter $x_2$ equal to its mean value (reduced); (b) and (d) Comparison of the probability density functions and the cumulative distribution functions at the final simulation time $T_{end}=100$; (c) Relative error in the estimated mean and standard deviation using the samples with the reduced number of uncertain input.}
  \label{fig:comparison_UQ_case1}
\end{figure}

Figures~\ref{fig:comparison_UQ_case1}~(b) and (d) show the probability density functions (PDFs) and the cumulative distribution functions (CDFs) of the uncertain model output $z$ at the final simulation time obtained using these two samples. There is a good match in the PDFs and CDFs with Kolmogorov–Smirnov (K-S) two sample test shows the K-S distance nearly $3.6\cdot10^{-4}$ and $p$-value larger than $0.5$, therefore, the hypothesis that the two samples are drawn from the same distributions cannot be rejected\footnote{This conclusion also applies to the other simulation times (data not shown).}.
\end{example}

\subsubsection{Case 2}\label{sec:case_2}
We consider the linear case, where the sampling function $G:\R^{2} \to \R$ is given by $G(u, v) = u + v$.

\begin{theorem} \label{thm:case_2} Let $g:(0, 1)^{n + m} \to \R$ be a function in $L^{2}((0, 1)^{n + m})$ such that
\begin{equation*}
  \label{eq:case3}
  g(x, \xi) = f(x) + h(\xi),
\end{equation*}
for some $f:(0, 1)^{n} \to \R$ and $h: (0, 1)^{m} \to \R$ satisfying $f \in L^{2}((0, 1)^{n})$ and $h \in L^{2}((0, 1)^{m})$. Then, we have
\begin{equation} \label{eq:STf_case2} S^g_{T_{x_i}} = \mu_{f, h} S_{T_{x_{i}}}^{f}, \end{equation} 
where
\begin{equation*} \mu_{f, h} := \frac{1}{1 + \frac{\Var(h)}{\Var(f)}}. \end{equation*} 
In particular, $S^{g}_{T_{x_{i}}} \le S_{T_{x_{i}}}^{f}$. 
\end{theorem}

\begin{proof}
The total SI of the input $x_i$ for the results of the model $g$ is equal to
\begin{align*}
\begin{split}
  S^g_{T_{x_i}}  &= \frac{\int (f(x) + h(\xi))^2 dx d\xi - \int (\int f(x) + h(\xi) dx_i )^2 dx_{\sim i} d\xi} 
                         {\int (f(x) + h(\xi))^2 dx d\xi - (f_{0} + h_{0})^2} 
                         \\ 
                 &= \frac{\int f^2(x) dx - \int (\int f(x) dx_i )^2 dx_{\sim i}} 
                         {\int f^2(x) dx  - f_0^2 + \int h^2(\xi)d\xi- h_0^2},
\end{split} 
\end{align*}
from which we get \eqref{eq:STf_case2} by dividing by $\Var(f)$ numerator and denominator.

Clearly, $\mu_{f, h} \in (0, 1]$, and so we conclude that $S^{g}_{T_{x_{i}}} \le S_{T_{x_{i}}}^{f}$. 

\end{proof}

Therefore, if the parameter $x_i$ is unimportant for $f$, it can be equated to its mean value in the uncertainty estimation of the model $g$.

\begin{example}[Standard Ornstein-Uhlenbeck process]

An example of Case 2 can be a multiscale model whose micro scale dynamics does not depend on the macro scale response. Let us consider the system (Figure~\ref{fig:results_case2}~(a)) \cite{Weinan_2011,weinan2005analysis}:
\begin{align*}
\begin{split}
\frac{\partial z}{\partial t} &= v + f(x),\\
\frac{\partial v}{\partial t} &= -\frac{1}{\epsilon} v + \frac{1}{\sqrt{\epsilon}} \dot{W}_t, \\
f(x) &= - x_1 + (x_2^2 \, x_3 + x_4),
\end{split} 
\end{align*}
where $z$ simulates the slow processes with $z(t=0) = 1$, $v$ is the fast process with $v(t=0)=1$, $\epsilon=10^{-2}$, $\dot{W}_t$ is a white noise with unite variance. The fast dynamics is the standard Ornstein-Uhlenbeck process. At any simulation time $t, \dot{W}_{t}$ plays the role of $\xi$ in Theorem~\ref{thm:case_2}. The macro model uncertain parameters $x=(x_1, x_2, x_3, x_4)$ follow normal distribution, such that
$x_1\sim \mathcal{N}(0, 10^{-4})$,
$x_2\sim \mathcal{N}(0, 2.5 \cdot 10^{-4})$, 
$x_3\sim \mathcal{N}(0, 2.5 \cdot 10^{-6})$,
$x_4\sim \mathcal{N}(0, 2.5 \cdot 10^{-6})$. 
The system is simulated using the forward Euler method with the macro time step $\Delta t_M = 1$ and the micro time step $\Delta t_{\mu} = 10^{-2}$. 

\begin{figure}
  \centering
  \includegraphics[width=\linewidth]{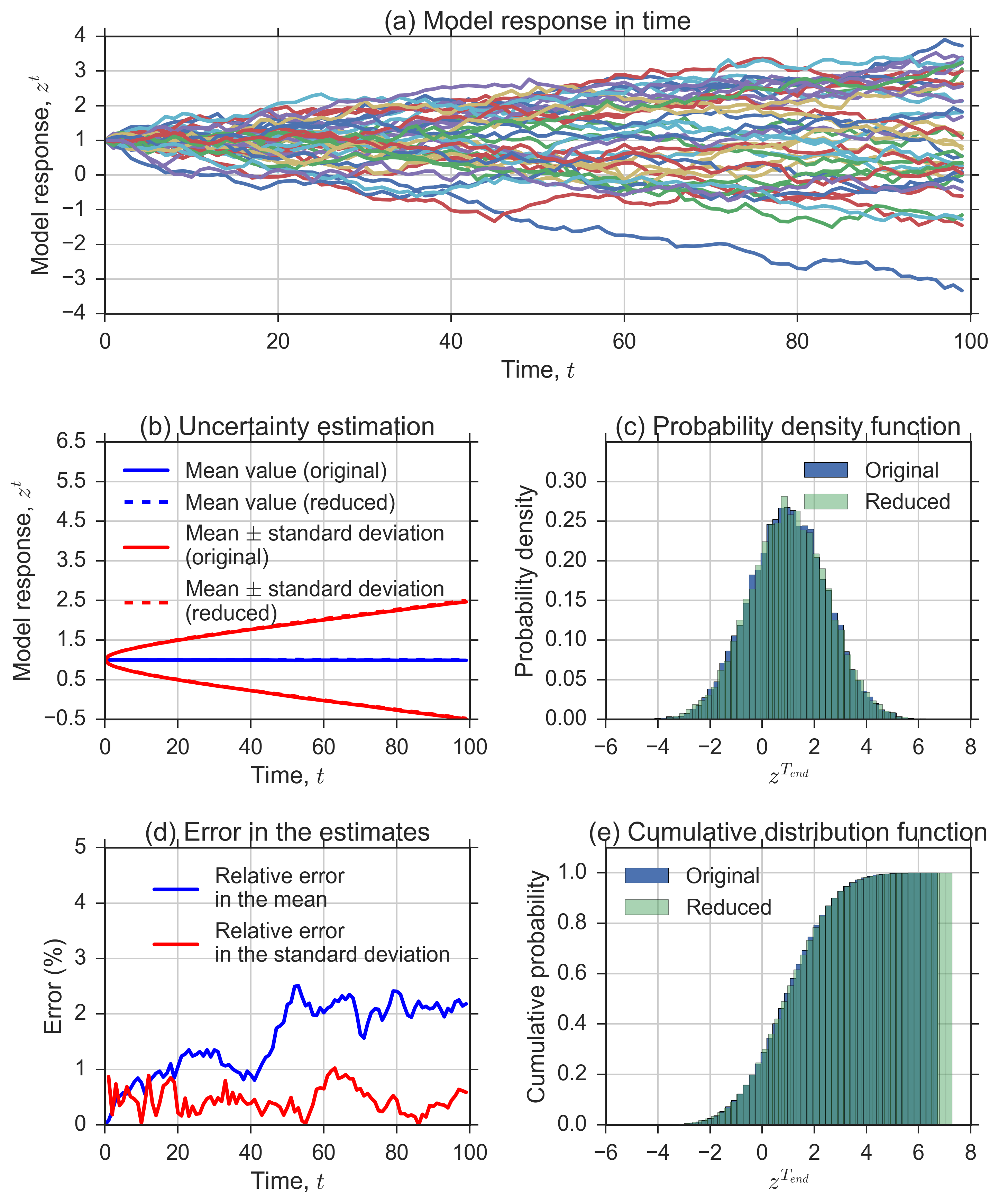}
  \caption{(a) Standard Ornstein-Uhlenbeck process; (b) Comparison of UQ result using the original sample and the sample obtained with values of the unimportant parameters $x_2$ and $x_3$ equal to their mean (reduced); (c) and (e) Comparison of the PDF and CDF at the final time step; (d) Relative error in the estimation of the mean and standard deviation.}
  \label{fig:results_case2}
\end{figure}

Sensitivity analysis of the function $f(x)$ yields
\begin{align*}
\begin{split}
S^f_{T_{x_1}} &\approx 7.7 \cdot 10^{-1},\\
S^f_{T_{x_2}} &\approx 2.6 \cdot 10^{-4},\\
S^f_{T_{x_3}} &\approx 3.9 \cdot 10^{-4},\\
S^f_{T_{x_4}} &\approx 2.0 \cdot 10^{-1}.
\end{split} 
\end{align*}
At any simulation time, the inputs $x_2$ and $x_3$ do not influence significantly the output of the function $f$. Therefore, they can be equated to their mean values without a substantial loss of accuracy of the uncertainty estimate as a consequence of Theorem~\ref{thm:case_2}.

The uncertainty estimation results of $z$ are presented in Figure~\ref{fig:results_case2}~(b). As it is proven analytically, the estimates obtained by sampling the model results with uncertain parameters $x_2$ and $x_3$ equal to their mean values are close to those resulting from samples where all the uncertain inputs vary. At any simulation time, the relative error between these estimates of the standard deviation does not exceed $1.1\%$ (Figure~\ref{fig:results_case2}~(d)). Additionally, Levene's test shows $p$-value about $0.66$, therefore, we cannot reject the hypothesis that the two samples are drawn from distributions with the same variance.

The PDFs and CDFs for the model result at the final time point obtained from these two samples are in Figure~\ref{fig:results_case2}~(c) and (e). There is a good match of the PDFs and CDFs obtained from these two samples, and K-S test produces the distance about $0.01$ and $p$-value about $0.47$, therefore, the hypothesis that the two samples are drawn from the same distributions cannot be rejected.
\end{example}

Some additional cases of the function $G$ for which the method of eliminating unimportant parameters to reduce the input dimensionality is valid are presented in the Appendix. 

\subsubsection{Counterexample}\label{sec:Counterexamples}

In this section, the importance of the examination of properties of the function $G$ is demonstrated. The counterexample illustrates that low sensitivity to a parameter of the response of a function $f$ does not necessarily imply low sensitivity to this parameter of a response of the function $g$. 

\begin{example}[Total sensitivity indices of composite functions]

Let $n = 2, m = 1$ and $i = 2$,
\begin{align}
\frac {\partial z}{\partial x_1} = -\frac{1}{4\sqrt[4]{\beta}} |x_1 + x_2 + \xi|^{-\frac{5}{4}},
\label{eq:counterexample_1}
\end{align}
for $(x_1, x_2, \xi) \in (0, 1)^3$, with $z(0, x_2, \xi) = \frac{1}{\sqrt[4]{\beta |x_2 + \xi|}}$ and $\beta > 0$ some fixed parameter. The solution to equation~\eqref{eq:counterexample_1} can be represented using the following system 
\begin{align*}
u=f(x) & = x_{1} + \beta x_{2},\\
v=h(x_{1}, \xi) & = (1 - \beta) x_{1} - \beta \xi, \\
G(u, v) &= \frac{1}{\sqrt[4]{|u - v|}},
\end{align*}
so that
\begin{equation*}
z(x, \xi) = g(x, \xi) = \frac{1}{\sqrt[4]{\beta}} \frac{1}{\sqrt[4]{x_{1} + x_{2} + \xi}}.
\end{equation*}
Let us now directly obtain sensitivity indices of the function $f(x)$ for the parameter $x_2$:
\begin{align*}
\begin{split}
S^{f}_{T_{x_{2}}} = \frac{\frac{1}{3} + \frac{\beta^{2}}{3} + \frac{\beta}{2} - \left( \frac{1}{3} + \frac{\beta^{2}}{4} + \frac{\beta}{2} \right)}{\frac{1}{3} + \frac{\beta^{2}}{3} + \frac{\beta}{2} - \left( \frac{\beta + 1}{2} \right)^2}
= \frac{\frac{\beta^{2}}{12}}{\frac{\beta^{2} + 1}{12}} = \frac{\beta^{2}}{1 + \beta^{2}}.
\end{split}
\end{align*}
Note that $S^{f}_{T_{x_{2}}}$ can be made arbitrarily small as $\beta \to 0$: for instance, by choosing $\beta \in \left (0, \frac{1}{10} \right )$, we get
\begin{equation*}
S^{f}_{T_{x_{2}}} < \frac{1}{100},
\end{equation*}
so that $x_{2}$ becomes an unimportant input for $f$.

On the other hand, sensitivity of the function $g(x, \xi)$ does not depend on $\beta$:
\begin{align*}
\begin{split}
S^{g}_{T_{x_{2}}} &= \frac
{\frac{1}{\sqrt{\beta}}\int_{(0, 1)^3} \frac{1}{
\sqrt{x_{1} + x_{2} + \xi}} dx_1 dx_2 d\xi-
\frac{1}{\sqrt{\beta}} \int_{(0, 1)^2} \left(\int_{(0, 1)} \frac{1}{\sqrt[4]{x_{1} + x_{2} + \xi}}dx_2\right)^2 dx_1d\xi}
{\frac{1}{\sqrt{\beta}}\int_{(0, 1)^3} \frac{1}{\sqrt{x_{1} + x_{2} + \xi}} dx_1 dx_2 d\xi-
\frac{1}{\sqrt{\beta}}\left( \int_{(0, 1)^3} \frac{1}{\sqrt[4]{x_{1} + x_{2} + \xi}} dx_1 dx_2 d\xi \right)^2}\\
&=\frac{\int_{(0, 1)^3} \frac{1}{\sqrt{x_{1} + x_{2} + \xi}} dx_1 dx_2 d\xi-
 \int_{(0, 1)^2} \left(\int_{(0, 1)} \frac{1}{\sqrt[4]{x_{1} + x_{2} + \xi}}dx_2\right)^2 dx_1d\xi}
{\int_{(0, 1)^3} \frac{1}{\sqrt{x_{1} + x_{2} + \xi}} dx_1 dx_2 d\xi-
\left( \int_{(0, 1)^3} \frac{1}{\sqrt[4]{x_{1} + x_{2} + \xi}} dx_1 dx_2 d\xi \right)^2}.
\end{split}
\end{align*}
In addition, since $g$ is symmetrical, 
\begin{equation*}
S^g_{T_{x_{2}}} = S^g_{T_{x_{1}}} = S^g_{T_{\xi}}.
\end{equation*} 
Hence, this proves that $x_{2}$ is not an unimportant input for the function $g$, since it must be as relevant as $x_{1}$ and $\xi$. Therefore, in general, it is wrong to eliminate an uncertain input from UQ only based on sensitivity analysis of a single scale model without verifying that $S^g_{T_{x_{i}}} \leq \lambda S^f_{T_{x_{i}}}$ holds for some finite $\lambda \ge 0$ as in Theorem~\ref{thm:case_1} and Theorem~\ref{thm:case_2}.
\end{example}

\section{Concluding remarks}\label{sec:Conclusions}

An application of sensitivity analysis to reduce dimensionality of multiscale models in order to improve the performance of their uncertainty estimation is discussed in this paper. It has been shown that for some multiscale models, the estimates of Sobol sensitivity indices of a single scale output can be used as an estimate of the upper bound for the sensitivity of the output of the whole multiscale model. In other words, knowledge on the importance of inputs from single scale models can be used to find the effective dimensionality of the overall multiscale model. Two classes of coupling function $G$ (multiplicative, additive) were considered, where the approach was demonstrated to work, based on Theorems~\ref{thm:case_1} and~\ref{thm:case_2}, and two examples. However, a counterexample was also constructed, showing that the success of the method strongly depends on the properties of the coupling function $G$. Obviously, this analysis only covers a very small portion of possible coupling functions, and a more systematic or case by case investigation would be warranted.

The next step is to apply the proposed approach to real-world multiscale applications, for instance, to a multiscale fusion model \cite{luk2019compat} and to a coupled human heart model \cite{Santiago2018}. Uncertainty quantification applied to these models is computationally expensive due to the high dimension of the model parameters. Therefore, the SA analysis on single scale models to reduce the dimensionality of the overall multiscale model input can be one of the possible ways to improve the efficiency of the model uncertainty quantification.

\section*{Funding}
This work is a part of the eMUSC (Enhancing Multiscale Computing with Sensitivity Analysis and Uncertainty Quantification) project. A. N. and A. H. gratefully acknowledge financial support from the Netherlands eScience Center. This project has received funding from the European Union Horizon 2020 research and innovation programme under grant agreement \#800925 (VECMA project). 


\section*{Declarations of interest: none}
None.

\appendix\label{sec:appendix}
\section*{Appendices}
\addcontentsline{toc}{section}{Appendices}
\renewcommand{\thesubsection}{\Alph{subsection}}

In this Appendix, additional cases of the function $G$ are considered. In particular, relations between the function $f$ and two or more functions representing the micro model are investigated, in this way allowing for vector valued functions $h$. Overall, our goal here is to show that the method presented in this work can be applied to different types of functions of the multiscale model components.

\section{Case 3} \label{sec:case_3}

Consider the affine linear case: $G: \R^{3} \to \R$ given by $G(u, v_{1}, v_{2}) = u v_{1} + v_{2}.$

\setcounter{theorem}{0}
\renewcommand{\thetheorem}{\Alph{section}\arabic{theorem}}
\begin{theorem} \label{thm:case_3} Let $g: (0, 1)^{m + n + k} \to \R$ be a function in $L^{2}((0, 1)^{m + n + k})$ such that
\begin{equation*}\label{eq:case4}
  g(x, \xi, \eta) = f(x)h_{1}(\xi) + h_{2}(x_{\sim i}, \eta),
\end{equation*}
for some $f: (0, 1)^{n} \to \R$, $h_{1}: (0, 1)^{m} \to \R$ and $h_{2}: (0, 1)^{k + n - 1} \to \R$ satisfying $f \in L^{2}((0, 1)^{n})$, $h_{1} \in L^{2}((0, 1)^{m})$, and $h_{2} \in L^{2}((0, 1)^{k + n - 1})$. Then,
\begin{equation} \label{eq:STf_case4_th} S^{g}_{T_{x_{i}}} = \gamma_{f, h_{1}, h_{2}} S^{f}_{T_{x_{i}}}, \end{equation}
where
\begin{equation*} \label{gamma_f_h_varphi_def} 
\gamma_{f, h_{1}, h_{2}} := \frac{\int f^2(x) dx - f_{0}^{2}} 
  {\int f^2(x) dx - \frac{f^2_0 (h_{1})^2_0}{\int h_{1}^2(\xi) d\xi}+ \frac{\int \int h_{2}^2(x_{\sim i}, \eta) dx_{\sim i} d\eta - (h_{2})_0^2}{\int h_{1}^2(\xi) d\xi} + 2 (h_{1})_0 \frac{ (f h_{2})_{0}-f_{0} (h_{2})_{0}}{\int h_{1}^2(\xi) d\xi}}. 
 \end{equation*}
If, additionally, it is assumed that
\begin{equation} \label{covariance_f_varphi} (h_{1})_{0}(f h_{2})_{0} \geq f_{0}(h_{1})_{0}(h_{2})_{0}, \end{equation} 
then $$S^g_{T_{x_i}} \leq S^f_{T_{x_i}}.$$ 
\end{theorem}

\begin{proof}
We compute
\begin{align*}
\begin{split}
\Var(g)_{T_{x_i}} =& \int \int f^2(x) h_{1}^2(\xi) dx d\xi + \int \int h_{2}^2(x_{\sim i}, \eta) dx_{\sim i} d\eta + 2 (h_{1})_0 (f h_{2})_{0} +  \\ & 
                    - \int \int \left(\int f(x) h_{1}(\xi) dx_i \right)^2 dx_{\sim i} d\xi + \\ &
                    - \int \int h_{2}^2(x_{\sim i}, \eta) dx_{\sim i} d\eta - 2 (h_{1})_0 (f h_{2})_{0}, \\
\Var(g) =& \int \int f^2(x) h_{1}^2(\xi) dx d\xi + \int \int h_{2}^2(x_{\sim i}, \eta) dx_{\sim i} d\eta + 2 (h_{1})_0 (f h_{2})_{0} + \\&
                    - h_{0}^{2} f_{0}^{2} - (h_{2})_{0}^{2} - 2 f_{0} (h_{1})_{0} (h_{2})_{0}.
\end{split} 
\end{align*}
Thus, the total SI of the input $x_i$ for the results of the model $g(x, \xi, \eta)$ is equal to
\begin{align*}
\begin{split}
  \label{eq:STf_case3}
  S^g_{T_{x_i}}  = &\frac{\Var(g)_{T_{x_i}}}{\Var(g)}
  \\ 
  = &  \frac{\int f^2(x) dx - \int (\int f(x) dx_i )^2 dx_{\sim i}} 
  {\int f^2(x) dx - \frac{f^2_0 (h_{1})^2_0}{\int h_{1}^2(\xi) d\xi}+ \frac{\int \int h_{2}^2(x_{\sim i}, \eta) dx_{\sim i} d\eta- (h_{2})_0^2}{\int h_{1}^2(\xi) d\xi} + 2 (h_{1})_0 \frac{ (f h_{2})_{0}-f_{0}(h_{2})_{0}}{\int h_{1}^2(\xi) d\xi}},
\end{split} 
\end{align*}
from which \eqref{eq:STf_case4_th} follows. By Cauchy-Schwarz inequality, we have
\begin{align*} (h_{1})^2_0 & \leq \int h_{1}^2(\xi) d\xi, \\
(h_{2})_{0}^{2} & \leq \int \int h_{2}^2(x_{\sim i}, \eta) dx_{\sim i} d\eta, \end{align*} 
which imply
\begin{equation*} \gamma_{f, h_{1}, h_{2}} \le \frac{\Var(f)} 
  {\Var(f) + 2 (h_{1})_0 \frac{ (f h_{2})_{0}-f_{0}(h_{2})_{0}}{\int h_{1}^2(\xi) d\xi}}. 
  \end{equation*}
To estimate the last term at the denominator, \eqref{covariance_f_varphi} is employed, yielding $$\gamma_{f, h_{1}, h_{2}} \le 1,$$ and the result follows.
\end{proof}

Note that, in the previous theorem, $h_{2}$ can be independent of more than one input $x_j$, however, it is crucial to assume the independence from the unimportant parameters which we want to exclude from uncertainty quantification.

\begin{remark}
It is noticed that condition \eqref{covariance_f_varphi} is equivalent to assume that $\E(h_{1}) \Cov(f, h_{2}) \ge 0$, since $\Cov(f, h_{2}) = (f h_{2})_{0} - f_{0} (h_{2})_{0}$.
Under the same assumption, one can get the following lower bound on $\gamma_{f, h_{1}, h_{2}}$:
\begin{equation*} \gamma_{f, h_{1}, h_{2}} \ge \frac{\int f^2(x) dx - f_{0}^{2}} 
  {\int f^2(x) dx + \frac{\int \int h_{2}^2(x_{\sim i}, \eta) dx_{\sim i} d\eta - (h_{2})_0^2}{\int h_{1}^2(\xi) d\xi} + 2 (h_{1})_0 \frac{ (f h_{2})_{0}-f_{0}(h_{2})_{0}}{\int h_{1}^2(\xi) d\xi}}. \end{equation*}
On the other hand, if $\E(h_{1}) \Cov(f, h_{2}) \le 0$; that is, $(h_{1})_{0}(f h_{2})_{0}  \le f_{0}(h_{1})_{0} (h_{2})_{0}$, then
\begin{equation*} \gamma_{f, h_{1}, h_{2}} \ge \frac{\int f^2(x) dx - f_{0}^{2}} 
  {\int f^2(x) dx + \frac{\int \int h_{2}^2(x_{\sim i}, \eta) dx_{\sim i} d\eta - (h_{2})_0^2}{\int h_{1}^2(\xi) d\xi}}. \end{equation*}
In addition, if it is assumed that $(h_{1})_{0} (f h_{2})_{0}  \le f_{0} (h_{1})_{0} (h_{2})_{0}$ and that $$\Var(f) + \frac{\Var(h_{2})}{\int h_{1}^{2}(\xi) d \xi} + \Cov(f, h_{2}) \ge 0,$$ we obtain the following upper bound for $\gamma_{f, h_{1}, h_{2}}$:
\begin{equation*} \gamma_{f, h_{1}, h_{2}} \le \frac{1} 
  {1 + \frac{\int \int h_{2}^2(x_{\sim i}, \eta) dx_{\sim i} d\eta - (h_{2})_0^2}{\Var(f) \int h_{1}^2(\xi) d\xi} + 2 (h_{1})_0 \frac{ (f h_{2})_{0}-f_{0}(h_{2})_{0}}{\Var(f) \int h_{1}^2(\xi) d\xi}}. \end{equation*}
\end{remark}

\section{Case 4}

A variant of the linear case $G(u, v) = u + v$ is considered. The difference with Case 2 of Theorem \ref{thm:case_2} is that now the functions $f$ and $h$ depend on the same set of variables.

\begin{theorem} \label{thm:case_4} Let $g: (0, 1)^{n} \to \R$ be a function in $L^{2}((0, 1)^{n})$ such that
\begin{equation*}
\label{eq:res_macro_t}
  g(x) = f(x) + h(x),
\end{equation*}
for some $f, h: (0, 1)^{n} \to \R$, $f, h \in L^{2}((0, 1)^{n})$. Then, if $$\Cov(f, h) = (f h)_{0} - f_{0} h_{0} \geq 0,$$ 
we have
\begin{equation} 
\label{eq:case_5_bound}
S^{g}_{T_{x_{i}}} \le \frac{ \left ( \sqrt{ S^{f}_{T_{x_{i}}} \Var(f)} + \sqrt{ S^{h}_{T_{x_{i}}} \Var(h)} \right )^{2}}{\Var(f) + \Var(h)},
\end{equation}
and so
\begin{equation} \label{eq:case_5_bound_easy} S^{g}_{T_{x_{i}}} \le 2 \max \{S^{f}_{T_{x_{i}}}, S^{h}_{T_{x_{i}}} \}, \end{equation}
where the factor $2$ is sharp.
\end{theorem}

\begin{proof}
By a simple computation, it follows that
\begin{align*}
\begin{split}
    S^{g}_{T_{x_{i}}} & = \frac{\int (f+h)^2 dx - \int (\int (f + h) dx_i)^2 dx_{\sim i}} {\int (f+h)^2 dx -(f_{0} + h_{0})^2}  
    \\ 
    \\
    & =\frac{\Var(f)_{T_{x_i}} + \Var(h)_{T_{x_i}} + 2 \int fh dx - 2 \int(\int f dx_i) (\int h dx_i) dx_{\sim i}}
    {\Var(f) + \Var(h) + 2\Cov(f, h)},
\end{split} 
\end{align*}
where $\Var(f)_{T_{x_{i}}} = \int f^{2}(x) \, dx - \int (\int f(x) \, d x_{i})^{2} \, d x_{\sim i}$, and $\Var(h)_{T_{x_{i}}}$ is defined analogously.
Then, by applying the Cauchy-Schwarz inequality to the functions $f(x) - \int f(x) \, d x_{i}$ and $h(x) - \int h(x) \, d x_{i}$, we get
\begin{equation} \label{Cov_x_i} \left | \int fh dx - \int \left (\int f dx_i \right ) \left ( \int h dx_i \right ) dx_{\sim i} \right | \le \sqrt{\Var(f)_{T_{x_i}} \Var(h)_{T_{x_i}} }. \end{equation}
Thus, if $\Cov(f, h) \geq 0$, by \eqref{Cov_x_i} we obtain
\begin{equation*} 
S^{g}_{T_{x_{i}}} \le \frac{ S^{f}_{T_{x_{i}}} \Var(f) + S^{h}_{T_{x_{i}}} \Var(h) + 2 \sqrt{ S^{f}_{T_{x_{i}}} \Var(f) S^{h}_{T_{x_{i}}} \Var(h)} }{\Var(f) + \Var(h)},
\end{equation*}
from which \eqref{eq:case_5_bound} immediately follows. Finally, we show that
\begin{equation} \label{eq:algebraic_inequality} \frac{(\sqrt{a y} + \sqrt{b z})^{2}}{a + b} \le 2 \max \{y, z\}
\end{equation}
for any $a, b, y, z > 0$. Indeed, without loss of generality, let $y > z$, and recall that $2 \sqrt{ab} \le a + b$: then, 
\begin{equation*} \frac{(\sqrt{a y} + \sqrt{b z})^{2}}{a + b} \le y \frac{(\sqrt{a} + \sqrt{b})^{2}}{a + b} = y \frac{a + b + 2 \sqrt{ab}}{a + b} \le 2 y.
\end{equation*}
Moreover, the factor $2$ is sharp: if $y = z$ and $a = b$, 
\begin{equation*} \frac{(\sqrt{a y} + \sqrt{b z})^{2}}{a + b} = y \frac{4 a}{2a} = 2 y = 2 \max\{y, z\}.
\end{equation*}
Therefore, inequality \eqref{eq:algebraic_inequality} shows that \eqref{eq:case_5_bound} implies \eqref{eq:case_5_bound_easy}.
\end{proof}

The bound given by \eqref{eq:case_5_bound} means that the total sensitivity index $S^{g}_{T_{x_{i}}}$ for the function $g$ of the input $x_i$ is controlled by $S^{f}_{T_{x_{i}}}$ and $S^{h}_{T_{x_{i}}}$. 
It is clear that this result can be applied also to a function $g$ of the form 
\begin{equation*} g(x) = f(x) + h_{1}(x) + h_{2}(x) + \dots + h_{k}(x),
\end{equation*}
for any $k \ge 1$. Indeed, it is enough to proceed by iteration: at first, we let $$h(x) = h_{1}(x) + h_{2}(x) + \dots + h_{k}(x),$$
then \eqref{eq:case_5_bound} is applied to $S^{h}_{T_{x_{i}}}$, by seeing $h$ as $$h(x) = h_{1}(x) + \tilde{h}_{2}(x),$$
where
$$\tilde{h}_{2}(x) = h_{2}(x) + \dots + h_{k}(x).$$
By applying this procedure $k$ times, the desired result is obtained. 
However, since the factor $2$ in \eqref{eq:case_5_bound_easy} is sharp, in general we cannot hope to obtain a better control than
\begin{equation*} \label{eq:case_5_bound_easy_iterated} S^{g}_{T_{x_{i}}} \le 2^{k} \max \{S^{f}_{T_{x_{i}}}, S^{h_{1}}_{T_{x_{i}}}, S^{h_{2}}_{T_{x_{i}}}, \dots, S^{h_{k}}_{T_{x_{i}}} \},
\end{equation*}
where the factor $2^{k}$ is again sharp.

\section{Case 5}

A variant of Case 3 (Theorem \ref{thm:case_3}), $G(u, v_{1}, v_{2}) = uv_{1} + v_{2}$ is considered.  This time, we assume dependence of $h_{2}$ also on the input $x_{i}$.
 
\begin{theorem} \label{thm:case_5} Let $g : (0, 1)^{n + m + k} \to \R$ be a function in $L^{2}((0, 1)^{n + m + k})$ such that
\begin{align*}
\begin{split}
  g(x, \xi, \eta) = f(x) h_{1}(\xi) + h_{2}(x,\eta)
\end{split} 
\end{align*}
for some $f \in L^{2}((0, 1)^{n}), h_{1} \in L^{2}((0, 1)^{m})$ and $h_{2} \in L^{2}((0, 1)^{n + k})$. Then, if $\E(h_{1}) \Cov(f, h_{2}) \ge 0$, 
\begin{equation} \label{eq:STf_case4} S^{g}_{T_{x_{i}}} \le  \frac{  S^{f}_{T_{x_{i}}} (\int h_{1}^2(\xi) \, d\xi ) \Var(f)  + S^{h_{2}}_{T_{x_{i}}} \Var(h_{2})  + 2 (h_{1})_{0} \sqrt{S^{f}_{T_{x_{i}}} \Var(f) S^{h_{2}}_{T_{x_{i}}} \Var(h_{2})}}
    {( \int h_{1}^2(\xi) \, d\xi ) \Var(f) + f_{0}^{2} \Var(h_{1}) + \Var(h_{2})}. \end{equation}
\end{theorem}
\begin{proof} It is enough to evaluate $S^{g}_{T_{x_{i}}}$. We have
\begin{align*}
& \int g^2(x, \xi, \eta)  \, dx \, d\xi \, d \eta - \int \left (\int g(x, \xi, \eta) \, dx_i \right )^2 \, dx_{\sim i} \, d\xi \, d \eta \\
& = \int \int (f(x) h_{1}(\xi) + h_{2}(x, \eta))^2 \, dx \, d\xi \, d \eta + \\
& - \int\int \left (\int (f(x) h_{1}(\xi) + h_{2}(x, \eta)) \, dx_i \right )^2 \, dx_{\sim i} \, d\xi \, d \eta \\
    & = \left ( \int h_{1}^2(\xi) \, d\xi \right ) \Var(f)_{T_{x_i}} +  \Var(h_{2})_{T_{x_i}} +\\
    & + 2 (h_{1})_{0} \left ( \int f(x) h_{2}(x, \eta) \, dx \, d \eta - 2 \int \left (\int f(x) \, dx_i \right) \left (\int h_{2}(x, \eta) \, dx_i \right ) \, dx_{\sim i} \, d \eta \right ) \\
    & \le \left ( \int h_{1}^2(\xi) \, d\xi \right ) \Var(f)_{T_{x_i}} +  \Var(h_{2})_{T_{x_i}} + 2 (h_{1})_{0} \sqrt{\Var(f)_{T_{x_{i}}} \Var(h_{2})_{T_{x_{i}}}},
\end{align*}
by \eqref{Cov_x_i}.
On the other hand, we get
\begin{align*}
\int \int g^{2}(x, \xi, \eta) \, dx \, d \xi \, d \eta - g_{0}^{2} & = \int \int (f(x) h_{1}(\xi) + h(x, \eta))^2 \, dx \, d\xi \, d \eta -(f h_{1} + h_{2})^2_0 \\
& = \left ( \int h_{1}^2 \, d\xi \right ) \Var(f) + f_{0}^{2} \Var(h_{1}) + \Var(h_{2}) + \\
& + 2 (h_{1})_{0} \Cov(f, h_{2}). \\
& \ge \left ( \int h_{1}^2 \, d\xi \right ) \Var(f) + f_{0}^{2} \Var(h_{1}) + \Var(h_{2}),
\end{align*}
since $(h_{1})_{0} \Cov(f, h_{2}) \ge 0$. Then, it follows that 
\begin{align*}
\begin{split}
    S^{g}_{T_{x_{i}}} & \le \frac{\left ( \int h_{1}^2(\xi) \, d\xi \right ) \Var(f)_{T_{x_i}} +  \Var(h_{2})_{T_{x_i}} + 2 (h_{1})_{0} \sqrt{\Var(f)_{T_{x_{i}}} \Var(h_{2})_{T_{x_{i}}}}}{\left ( \int h_{1}^2 \, d\xi \right ) \Var(f) + f_{0}^{2} \Var(h_{1}) + \Var(h_{2})} \\
    & = \frac{ S^{f}_{T_{x_{i}}} (\int h_{1}^2 \, d\xi ) \Var(f)  + S^{h_{2}}_{T_{x_{i}}} \Var(h_{2})  + 2 (h_{1})_{0} \sqrt{S^{f}_{T_{x_{i}}} \Var(f) S^{h_{2}}_{T_{x_{i}}} \Var(h_{2})}}{( \int h_{1}^2 \, d\xi ) \Var(f) + f_{0}^{2} \Var(h_{1}) + \Var(h_{2})},
\end{split} 
\end{align*}
which is \eqref{eq:STf_case4}.
\end{proof}

If $h_{1} \equiv 1$ and $k = 0$, there is no dependence on $\xi$ and $\eta$, and Theorem \ref{thm:case_5} implies Theorem \ref{thm:case_4} for the functions $f$ and $h_{2}$.

\section{An estimate on a general class of model functions}\label{sec:general_res}

Let $G : \R^{2} \to \R$ be such that there exist $L \ge c > 0$ satisfying
\begin{align} \label{Lipschitz} |G(u, v) - G(u_{0}, v)| & \le L |u - u_{0}|, \\
\label{coercive} |G(u, v)| & \ge c \sqrt{u^{2} + v^{2} } \end{align}
for any $u, u_{0}, v \in \R$, which means that $G$ is Lipschitz in $u$, uniformly in $v$, and that it is a coercive function.

\begin{theorem} \label{general_case_thm} Let $g: (0, 1)^{n + m} \to \R$ be a function in $L^{2}((0, 1)^{n + m})$ such that $g_{0} = 0$ and  
\begin{equation} \label{g_function_form} g(x, \xi) = G(f(x), h(x_{\sim i}, \xi)) \end{equation}
for some functions $f : (0, 1)^{n} \to \R$ and $h : (0, 1)^{n + m - 1} \to \R$ satisfying $f \in L^{2}((0, 1)^{n})$ and $h \in L^{2}((0, 1)^{n + m -1})$. Then,
\begin{equation} \label{general_upper_bound} S_{T_{x_{i}}}^{g} \le 2 \frac{L^{2}}{c^{2}} \frac{\Var(f)}{\Var(f) + \Var(h)} S_{T_{x_{i}}}^{f}. \end{equation}
\end{theorem}  
\begin{proof} By \eqref{Lipschitz} and \eqref{g_function_form}, it follows that $g(x, \xi) \in L^{2}((0,1)^{n + m})$. 
Since $g_{0} = 0$, \eqref{coercive} implies
\begin{align*} \Var(g) & = \int_{(0, 1)^{n + m}} g^{2}(x, \xi) \, dx \, d \xi \\
& \ge c^{2} \left ( \int_{(0, 1)^{n}} f^{2}(x) \, dx + \int_{(0, 1)^{n + m - 1}} h^{2}(x_{\sim i}, \xi) \, d x_{\sim i} \, d \xi \right ) \\
& \ge c^{2} \left ( \Var(f) + \Var(h)  \right ). \end{align*}
We further notice that, by Jensen inequality combined with \eqref{Lipschitz} and \eqref{g_function_form}, we get
\begin{align*} & \int_{(0, 1)^{n + m}} \left ( g(x_{i}, x_{\sim i}, \xi) - \int_{0}^{1} g(\tilde{x}_{i}, x_{\sim i}, \xi) \, d \tilde{x}_{i} \right )^{2} \, d x_{i} \, d x_{\sim i} \, d \xi \\
& \le \int_{(0, 1)^{n + m}} \int_{0}^{1} |g(x_{i}, x_{\sim i}, \xi) - g(\tilde{x}_{i}, x_{\sim i}, \xi)|^{2} \, d \tilde{x}_{i} \, d x_{i} \, d x_{\sim i} \, d \xi \\
& = \int_{(0, 1)^{n + m + 1}} |G(f(x_{i}, x_{\sim i}), h(x_{\sim i}, \xi)) - G(f(\tilde{x}_{i}, x_{\sim i}), h(x_{\sim i}, \xi))|^{2} \, d \tilde{x}_{i} \, d x_{i} \, d x_{\sim i} \, d \xi \\
& \le L^{2} \int_{(0, 1)^{n - 1}} \int_{0}^{1} \int_{0}^{1} |f(x_{i}, x_{\sim i}) - f(\tilde{x}_{i}, x_{\sim i})|^{2} \, d \tilde{x}_{i} \, d x_{i} \, d x_{\sim i} \\
& = 2 L^{2} \left ( \int_{(0, 1)^{n}} f^{2}(x) \, dx - \int_{(0, 1)^{n - 1}} \left ( \int_{0}^{1} f(x_{i}, x_{\sim i}) \, d x_{i} \right )^{2} \, d x_{\sim i} \right ). \end{align*}
Therefore, combining these two inequalities, \eqref{general_upper_bound} is obtained.
\end{proof}

We notice that we can replace the assumption $g_{0} = 0$ in Theorem \ref{general_case_thm} with a weaker one.

\begin{corollary} \label{general_case_cor} Let $g: (0, 1)^{m + n} \to \R$ be a function in $L^{2}((0, 1)^{m + n})$ as in \eqref{g_function_form}, with $f \in L^{2}((0, 1)^{n})$, $h \in L^{2}((0, 1)^{n + m -1})$. Then, if $$c^{2}(\Var(f) + \Var(h) ) - g_{0}^{2} \ge 0,$$  we have
\begin{equation*} \label{general_upper_bound_cor} 
S_{T_{x_{i}}}^{g} \le 2 \frac{L^{2} \Var(f)}{c^{2}(\Var(f) + \Var(h)) - g_{0}^{2}} S_{T_{x_{i}}}^{f}. \end{equation*}
\end{corollary}  
\begin{proof} The proof is the same of Theorem \ref{general_case_thm}, one needs just to subtract the term $g_{0}^{2}$ at the denominator.
\end{proof}

\begin{remark} 
It is not difficult to see that we could restate Theorem \ref{general_case_thm} and Corollary \ref{general_case_cor} for a function $G : \R \times \R^{l} \to \R$; that is, allowing $v$ to be a vector $(v_{1}, v_{2}, \dots, v_{l})$ in $\R^{l}$. The Lipschitz condition would not change, while the coercivity condition \eqref{coercive} would become
$$ |G(u, v)| \ge c \sqrt{u^{2} + |v|^{2}} = c \sqrt{u^{2} + v_{1}^{2} + v_{2}^{2} + \dots + v_{l}^{2}}.$$
This would allow to have not only one function $h$, but a family of $l$ different functions $h_{1}, h_{2}, \dots , h_{l}$, which could be seen as a vector valued function 
$$ h = (h_{1}, h_{2}, \dots, h_{l}) : (0, 1)^{n + m - 1} \to \R^{l},$$ 
satisfying 
$$\Var(h) = \Var(h_{1}) + \Var(h_{2}) + \dots + \Var(h_{l}).$$
\end{remark}

The next example illustrates that the admissible function $G$ for Theorem \ref{general_case_thm} and Corollary \ref{general_case_cor} can be very nonlinear.

\begin{example}
Let 
$$G(u, v) = a |u| + b |v| + \arctan{\left ( \frac{|u| + |v|}{1 + v^{2}} \right )},$$
for some $a, b > 0$. Then, $G$ is Lipschitz in $u$ uniformly in $v$, since
\begin{equation*}
    \frac{\partial G (u, v)}{\partial u} = \left ( a + \frac{(1 + v^{2})}{(1 + v^{2})^{2} + u^{2} + v^{2} + 2 |u||v|} \right ) {\rm sgn}(u),
\end{equation*}
which is a bounded function. Thus, $G$ satisfies condition \eqref{Lipschitz}, with
$$L = \sup_{u, v} \left | \frac{\partial G (u, v)}{\partial u} \right |.$$ 
As for the coercivity condition \eqref{coercive}, it is easy to see that $G(u, v) \ge 0$ and
\begin{equation*}
    G(u, v) \ge \min\{a, b\}(|u| + |v|) \ge \min\{a, b\} \sqrt{u^{2} + v^{2}},
\end{equation*}
so that we have $c = \min\{a, b\}$.
It is clear that, since $G(u, v) \ge 0$, any $g(x, \xi) = G(f(x), h(x_{\sim i}, \xi))$ cannot satisfy $g_{0} = 0$, unless $f = h = 0$. Hence, in general, we can apply Corollary \ref{general_case_cor} only if we ensure that
$$\min\{a, b\}^{2}(\Var(f) + \Var(h) ) \ge  g_{0}^{2}.$$
\end{example}


\bibliographystyle{elsarticle-num}
\bibliography{mybibfileUQSA}{}

\end{document}